\documentclass[conference]{IEEEtran}
\IEEEoverridecommandlockouts
\usepackage{cite}
\usepackage[fleqn]{amsmath}
\usepackage{amsthm}
\usepackage{algorithm}
\usepackage{algorithmicx}
\usepackage[noend]{algpseudocode}
\usepackage{graphicx}
\usepackage{textcomp}
\usepackage{xcolor}
\usepackage{multirow}
\usepackage{enumerate}

\title{Live Exploration with Mobile Robots in a Dynamic Ring, Revisited\thanks{A. R. Molla was supported, in part, by DST Inspire Faculty research grant DST/INSPIRE/04/2015/002801, Govt. of India. The work of William K. Moses Jr. was supported in part by a Technion fellowship.}\\
}
\author{\IEEEauthorblockN{Subhrangsu Mandal}
\IEEEauthorblockA{
\textit{Indian Institute of Technology Kharagpur}\\
Kharagpur, India \\
subhrangsum@cse.iitkgp.ac.in}
\and
\IEEEauthorblockN{Anisur Rahaman Molla}
\IEEEauthorblockA{
\textit{Indian Statistical Institute}\\
Kolkata, India \\
molla@isical.ac.in}
\and
\IEEEauthorblockN{William K. Moses Jr.}
\IEEEauthorblockA{
\textit{Technion - Israel Institute of Technology}\\
Haifa, Israel \\
wkmjr3@gmail.com}
}

\newtheorem{theorem}{Theorem}[section]
\newtheorem{corollary}{Corollary}[theorem]
\newtheorem{lemma}[theorem]{Lemma}
\newtheorem{observation}{Observation}
\newtheorem{remark}{Remark}

\newcommand{\R}{\mathcal{R}}

\newcommand{\subh}[1]{{\color{green}\bf [Subh: #1]}}

\begin{document}

\maketitle


\begin{abstract}
The graph exploration problem requires a group of mobile robots, initially placed arbitrarily on the nodes of a graph, to work collaboratively to explore the graph such that each node is eventually visited by at least one robot. One important requirement of  exploration is the {\em termination} condition, i.e., the robots must know that exploration is completed. The problem of live exploration of a dynamic ring using mobile robots was recently introduced in [Di Luna et al., ICDCS 2016]. In it, they proposed multiple algorithms to solve exploration in fully synchronous and semi-synchronous settings with various guarantees when $2$ robots were involved. They also provided guarantees that with certain assumptions, exploration of the ring using two robots was impossible. An important question left open was how the presence of $3$ robots would affect the results. In this paper, we try to settle this question in a fully synchronous setting and also show how to extend our results to a semi-synchronous setting. 

In particular, we present algorithms for exploration with explicit termination using $3$ robots in conjunction with either (i) unique IDs of the robots and edge crossing detection capability (i.e., two robots moving in opposite directions through an edge in the same round can detect each other), or (ii) access to randomness. The time complexity of our deterministic algorithm is asymptotically optimal. We also provide complementary impossibility results showing that there does not exist any explicit termination algorithm for $2$ robots even when each robot has a unique ID, edge crossing detection capability, and access to randomness. The theoretical analysis and comprehensive simulations of our algorithm show the effectiveness and efficiency of the algorithm in dynamic rings. We also present an algorithm to achieve exploration with partial termination using $3$ robots with unique IDs in the semi-synchronous setting, when robots have access to edge crossing detection capability and randomness but do not know a bound on the size of the ring or have access to a landmark or are guaranteed that robots have common chirality. Our algorithms are fully decentralized, lightweight, and easily implementable. 
\end{abstract}

\pagestyle{plain}



\section{Introduction}
\label{sec:intro}
The research area of autonomous mobile robots in a graph setting has been well studied over the years. Many fundamental problems have been studied in this area, such as the problem of exploration of a graph using multiple robots. In this problem, multiple robots are placed in nodes in the graph and the goal is to design an algorithm, run by each robot, such that all robots collectively visit each node at least once as quickly as possible. As this fundamental problem has been solved to a large degree in most vanilla settings \cite{AH00, DP99, DP14, FIPPP05, PP99}, its study has been extended to more exotic, but realistic, settings.

One such setting is a dynamic network. In the real world, dynamism is seen fairly regularly in networks. As with most things in real life, the dynamism that appears in the real world is quite complex. In order to work towards a deeper understanding of this complexity, we first start with a simpler model of dynamism, which is a restricted version of \textit{1-interval connectivity} applied to a ring in the synchronous setting. We now describe this version of dynamism. Consider $n$ vertices, $v_0, v_1, \ldots, v_{n-1}$, with an undirected edge between every node $v_i$ and $v_{i+1 \mod n}$, $\forall i: \, 0 \leq i \leq n-1$. In each round, the adversary can choose to remove at most one edge of the ring. 

In this setting, Di Luna et al.~\cite{DDFS16} were the first to study the problem of graph exploration when robots do not know what the adversary will do next (\textit{live} or \textit{online} dynamism). This is contrasted with another scenario, called \textit{post-mortem} dynamism, where robots have complete knowledge of how the adversary will control dynamism in every round. In~\cite{DDFS16}, Di Luna et al. studied both fully synchronous systems and semi-synchronous systems where nodes are anonymous, i.e. do not have unique IDs. In the fully synchronous setting, they show that by using just $2$ robots without unique IDs, subject to some assumptions, deterministic exploration of a ring in the presence of 1-interval connectivity is possible with termination detection. These assumptions include a mix of the following ideas: (i) robots have knowledge of the value of $n$, (ii) there exists a \textit{landmark} (a unique node that can be identified by robots as being unique), (iii) robots have common chirality (a common sense of clockwise/counterclockwise). They differentiate between \textit{explicit termination} detection where all robots can detect the completion of exploration and subsequently terminate, and \textit{partial termination} detection where at least one of the robots (but not necessarily all of them) detects completion and terminates. They show that deterministic ring exploration with explicit termination is possible with $2$ robots with the aforementioned assumptions. They also provide matching impossibility results that deterministic exploration with partial termination is impossible with $2$ robots when $n$ is unknown and no landmark is available, even in the presence of robots with unique IDs and common chirality. They also show that if $n$ is unknown, no landmark is available and the robots are anonymous, then regardless of the number of robots initially deployed on the ring, deterministic exploration with partial termination is impossible. This impossibility holds even if those robots have common chirality. An important question left unanswered was if exploration with $\geq 3$ robots is possible when no knowledge of $n$ is known and no landmark is available but robots may have IDs. In this paper, we further extended the exploration problem in the dynamic ring and try to settle this question in the fully synchronous setting and provide partial results in the semi-synchronous setting.

\subsection{Our Contributions}
In this paper, we look into exploration of a dynamic ring with $3$ robots and show various positive results when certain assumptions are made.

We show that deterministic exploration of a dynamic ring of size $n$ with explicit termination detection is indeed possible with $3$ robots when $n$ is unknown and no landmark is present. In fact, not only is exploration possible, but the running time of our algorithm (which is linear on the size of the ring) is asymptotically optimal. We require robots to have unique IDs and have the capability of edge crossing detection, i.e. two robots passing through the same edge in a given round in opposite directions can detect that they passed each other in that round. We also implement our algorithm and show that it outperforms the theoretical time bound for different parameter ranges.  

We subsequently remove the need for the edge crossing detection assumption with the help of randomness. We also show how to use randomness to remove the need for robots to have unique IDs. Note that this result when we achieve explicit termination with anonymous $3$ robots, no landmark, no knowledge of $n$, but access to randomness is in sharp contrast to the impossibility result of~\cite{DDFS16} where even partial termination with any number of robots is impossible without under the same setting but without access to randomness. We also show how to modify our algorithm to achieve partial termination with better runtime. Our positive results are summarized in Table~\ref{table:possibility-results}.

One may wonder if either the use of edge crossing detection or the use of randomness is sufficient for $2$ robots to bypass the impossibility result from~\cite{DDFS16}. We show that when robots only have access to edge crossing detection, exploration with partial termination of two robots is impossible. We further show that when the use of randomness is also allowed, exploration with explicit termination of two robots is impossible. Thus, we see that only with the use of $3$ robots do either of these capabilities provide sufficient power to overcome the impossibility of exploration with explicit termination. Our impossibility results are summarized in Table~\ref{table:impossibility-results} along with a comparison to the impossibility result from~\cite{DDFS16}.

Finally, we show how to use the ideas we built up throughout the paper in order achieve partial termination in the semi-synchronous setting when robots do not know an upper bound on the value of $n$ or a have access to a landmark. The algorithm uses $3$ robots with unique IDs, access to randomness, and access to the edge crossing detection capability.

\begin{table*}[ht]
	\caption{Fully synchronous setting, impossibility results.} 
	\centering 
		\resizebox{2.0\columnwidth}{!}{%
	\begin{tabular}{|c|c|c|c|c|}
		\hline
		Paper & \# of Robots & Assumptions & Even with Assumptions & Result \\
		\hline
		\hline
		\cite{DDFS16} & 2 & No knowledge of $n$, No landmark & Non-anonymous robots, Chirality & Partial termination impossible \\
		\hline
		\multirow{2}{*}{\cite{DDFS16}} & \multirow{2}{*}{Any} & No knowledge of $n$, No landmark, & \multirow{2}{*}{Chirality} & \multirow{2}{*}{Partial termination impossible} \\
		 &  & Anonymous robots &  &  \\
		\hline
		\multirow{2}{*}{Current paper} & \multirow{2}{*}{2} & \multirow{2}{*}{No knowledge of $n$, No landmark} & Non-anonymous robots, Chirality,  & \multirow{2}{*}{Partial termination impossible}\\
		& & & Edge crossing detection & \\
        \hline
		\multirow{2}{*}{Current paper} & \multirow{2}{*}{2} & \multirow{2}{*}{No knowledge of $n$, No landmark} & Non-anonymous robots, Chirality,  & \multirow{2}{*}{Explicit termination impossible}\\
		& & & Edge crossing detection, Access to randomness & \\
		\hline
	\end{tabular}
		}
	\label{table:impossibility-results}
\end{table*}

\begin{table*}[ht]
	\caption{Fully synchronous setting, possibility with 3 robots. Results: exploration with explicit termination.} 
	\centering 
		\resizebox{2.0\columnwidth}{!}{%
	\begin{tabular}{|c|c|}
		\hline
		Assumptions & Running time \\
		\hline
		\hline
		Non-anonymous robots, Edge crossing detection & Explicit termination in $O(n)$ rounds\\
        \hline
		Non-anonymous robots, Access to randomness & Explicit termination with probability $\geq 1 - 1/n$ in $O(n \log n)$ rounds on expectation \\
		\hline
		Access to randomness & Explicit termination with probability $\geq (1 - O(1/2^l))(1 - 1/n)$ in $O((n + n\cdot 2^l)\log n)$ rounds on expectation$^*$ \\
		\hline
		\multicolumn{2}{|l|}{$^*$ Here, $l$ is an input parameter to the algorithm.}\\ 
		\hline
	\end{tabular}
		}
	\label{table:possibility-results}
\end{table*}

\subsection{Related Work}
Exploration of static anonymous graphs using mobile robots has been studied for a very long time. A good survey on the topic is presented in~\cite{D19,Luna19}. Exploration on anonymous graphs with 1-interval connected dynamism is relatively new and the first paper to study it in the current model is~\cite{DDFS16}. In the paper, they look at exploration problem in a ring under 1-interval connected dynamism and provide various deterministic algorithms to solve the problem using $2$ robots for various assumptions. It should be noted that the way 1-interval connectivity is defined in their paper and also in the current paper is different from the original definition proposed in~\cite{KLO10, OW05}. Specifically, the original definition of 1-interval connectivity allows for permutations of the nodes of the graph, whereas in~\cite{DDFS16}, the nodes remain stationary and the adversary can only choose whether to remove at most one of a fixed set of edges.

A randomized approach to graph exploration was presented in~\cite{AKL08} via random walk, however the model of dynamism they look at is slightly different. 
Their approach is that of a lazy random walk, but when the rate of change of the graph is very fast, i.e. every round the adversary changes the graph, then things become complicated. Essentially their approach may take $\Omega(n^2)$ time to explore a dynamic ring of size $n$, however, it cannot guarantee any termination.  

There are other works in literature which have addressed the problem of exploration on dynamic graphs. Exploration problem on dynamic ring for T-interval connected case is addressed in~\cite{IW13}. They have addressed the problem in two scenarios. In one scenario the robot knows about all changes in the dynamic ring. In another case the robot has no knowledge about the changes but the edges are $\delta$-recurrent. They have extended their work in~\cite{IKW14} and addressed the exploration problem on cactus graph when change in the graph topology is known to the robot. There are other works like~\cite{EHK15, MS16} which address the exploration problem for general graphs in centralized environment when the change in the graph topology is already known. There are works~\cite{FMS13, IW11} which address the live or online version of the exploration problem in distributed environment for periodically varying graphs. In this case there are finite number of carriers in the graph and an edge between two nodes exist only when a carrier moves from one node to the other. They have assumed that the movement of each carrier is periodic. A very recent work \cite{GFMS19}, studies exploration in time-varying graphs (including $1$-interval connectivity) of arbitrary topology, investigates the number of robots necessary and sufficient to explore such graphs. 
There have been other papers that look at different problems such as gathering~\cite{DFPPSV18} and dispersion~\cite{AAMSS18} on dynamic graphs under 1-interval connectivity.



\subsection{Organization of Paper}
In Section~\ref{model}, we elaborate on the exact model of the system. In Section~\ref{sec:impossibility}, we present our impossibility results for termination with just 2 robots. In Section~\ref{explore}, we develop our algorithm to achieve exploration with explicit termination using 3 robots, including simulations of the algorithm. In Section~\ref{sec:randomness-exploration}, we show how to remove the requirement of edge crossing detection and unique IDs for robots in our algorithm through the creative application of access to randomness. In Section \ref{sec:ssync}, we extend our deterministic algorithm to achieve exploration with partial termination using 3 robots in semi-synchronous model. 
Finally, we conclude with future research directions in Section~\ref{sec:conclusions}.

\section{Network Model and Assumptions}
\label{model}
We consider a $1$-interval connected synchronous dynamic ring $\mathcal{R}$ of size $n$ as considered in \cite{Luna19,DDFS16}. As $\R$ is a ring, each node in $\R$ has two neighbours connected via two ports. The ring is anonymous, i.e., nodes are indistinguishable. We assume that the nodes are fixed, but the edges of $\mathcal{R}$ may change over time. More precisely, at any round one of the edges might be missing from $\R$. An adversary decides which edge to be deleted in a round. This dynamic ring is called as a $1$-interval connected ring \cite{KO11,DDFS16}. The adversary controls the edge deletion (and addition) with the knowledge of the algorithm and current states and positions of the robots. 

There are three robots $A = \{R_1, R_2, R_3\}$ which explore $\mathcal{R}$. Each robot is equipped with a finite memory, say $O(\log n)$ bits and computational capabilities. Each robot has a unique identifier (ID) and initially a robot only knows its own ID. Furthermore, we assume that the IDs are $k$-bit strings such that the length $k$ is $O(1)$. It is sufficient to represent $3$ distinct IDs with constant number of bits. ID of a robot is represented as $b_{k-1} b_{k-2} \cdots b_1 b_0$. We assume that the length of each ID is same. Initially robots do not know the size of the ring (not even any bound of it). The robots do not share any common chirality, i.e., the {\em clockwise} or {\em anti-clockwise} direction for all the robots may differ. During movement, at any node a robot can differentiate between the port through which it enters the node and the other port. All the robots execute the same protocol. Multiple robots can reside at a single node at the same time. The robots can move from one node to a neighboring node in some round if the corresponding edge is available in that round. A robot can {\em successfully move} towards a {\em fixed} direction if the corresponding adjacent edge is available in the dynamic ring; otherwise, if the edge is missing, the robot waits until the edge is available. We assume the {\em edge crossing detection}, i.e. two robots moving in opposite directions on the same edge in the same round, can detect that they passed each other in that round and exchange information.

We consider here a synchronous system which progresses in time steps, called as rounds. In a single round, the sequence of operations executed as follows: (i) the robots perform local computation and decide whether to move from the current node and the direction of the movement, (ii) the adversary removes at most one edge from the ring for this round, (iii) the robots execute their movements, if any, so long as the edge they wish to move over is present.

Note that we also consider a semi-synchronous system in Section~\ref{sec:ssync}. As most of the paper relates to a fully synchronous system as previously described, we postpone the description of the semi-synchronous system to Section~\ref{sec:ssync}.

\section{Impossibility of Exploration with 2 Robots}
\label{sec:impossibility}

In this section, we extend the impossibility results from~\cite{DDFS16} to the scenario where robots also have the edge crossing detection capability and access to randomness. First, we make a similar observation to Observation~2 from \cite{DDFS16}.

\begin{observation}\label{obs:never-meet}
The adversary can prevent two robots starting at different locations from meeting each other even if they have unlimited memory, common chirality, distinct known IDs, the edge crossing detection capability, and access to randomness.
\end{observation}

This observation is clear to see by considering the following strategy of an adversary. The adversary does nothing unless the following condition arises. If the two robots are on adjacent nodes, then the adversary removes the edge between those two nodes. Now, using this observation, we are able to prove the following impossibility result, which is a more general version of that seen in~\cite{DDFS16}. We note the nature of the proof is similar to that of Theorem~1 in~\cite{DDFS16}. 

\begin{theorem}
There does not exist any exploration algorithm with partial termination of anonymous rings of unknown size by two robots, even when robots have distinct IDs, common chirality, the edge crossing detection capability, and when the scheduler is fully synchronous.
\end{theorem}

\begin{proof}
Let there exist an algorithm $\mathcal{A}$ that achieves exploration with partial termination with two robots, say $R_1$ and $R_2$. Run this algorithm on a ring of size $n$ and consider any adversary strategy that prevents $R_1$ and $R_2$ from ever meeting. From Observation~\ref{obs:never-meet}, we know such a strategy is possible. Let us assume that, without loss of generality, $R_1$ terminates first after $T(n)$ rounds. We now construct a ring and adversary strategy such that $\mathcal{A}$ will fail, i.e. never achieve exploration.

Consider the ring of size $4T(n)+4$ and place $R_1$ and $R_2$ on nodes that are at distance $2T(n) + 2$ from each other, i.e. at opposite "ends" of the ring. Have each robot run $\mathcal{A}$ and let the adversary act such that $R_1$'s local view at each time step is similar to its view when $R_1$ ran $\mathcal{A}$ on the ring of size $n$. Thus, after $T(n)$ time steps (rounds), $R_1$ will terminate and by Observation~\ref{obs:never-meet}, $R_1$ never came into contact with $R_2$. Subsequently, in each future round the adversary will remove any edge that $R_2$ may want to traverse and thus ensure that $R_2$ does not explore any more nodes. After $T(n)$ rounds, $R_1$ and $R_2$ would have collectively explored at most $2T(n) + 2$ nodes and thus $\mathcal{A}$ fails to achieve exploration, which is a contradiction.
\end{proof}

We now provide a similar impossibility result when robots have access to randomness. Note that for this result, we are showing the impossibility of explicit termination and not partial termination. Also note that the proof is similar to that of the previous theorem with a few subtle but significant changes.

\begin{theorem}
There does not exist any exploration algorithm with explicit termination of anonymous rings of unknown size by two robots, even when robots have distinct IDs, common chirality, the edge crossing detection capability, access to randomness, and when the scheduler is fully synchronous.
\end{theorem}

\begin{proof}
Let there exist an algorithm $\mathcal{A}$ that achieves exploration with explicit termination with two robots, say $R_1$ and $R_2$. Run this algorithm on a ring of size $n$. Define $T(n)$ to be the maximum running time for both robots to terminate, over all choices of randomness and all adversarial strategies that prevented $R_1$ and $R_2$ from meeting. By Observation~\ref{obs:never-meet}, we know that such strategies exist. Define an execution $j$ of $\mathcal{A}$ as a vector of all the random choices, information communicated, local computation, and movements performed by both robots until termination. Define $V_{R_1}(i,j)$ as the vector of local views of $R_1$ for all rounds up to round $i$ for some execution $j$ of $\mathcal{A}$. Define $V_{R_1}(j)$ as the vector of local views of $R_1$ for all rounds up to termination for some execution $j$ of $\mathcal{A}$. Define $\mathcal{V_{R_1}}$ as the set of all $V_{R_1}(j)$ across all executions $j$ of $\mathcal{A}$ for all possible choices of randomness and all possible adversarial strategies subject to the condition that $R_1$ and $R_2$ never meet. We now construct a ring and adversary strategy such that $\mathcal{A}$ will fail, i.e. never achieve exploration.

Consider the ring of size $4T(n) + 4$ and place $R_1$ and $R_2$ on nodes that are at distance $2T(n) + 2$ from each other. Now, the adversary focuses on $R_1$ and acts so that the local view of $R_1$ at round $i$, $V_{R_1}(i,j)$ will always belong to $\mathcal{V_{R_1}}$ for the current execution $j$ of $\mathcal{A}$. This is possible because $R_1$ and $R_2$ never meet in any of the executions we considered and so the local view of $R_1$ in a given round is influenced only by which edge has been removed in that round. Furthermore, since $R_1$ and $R_2$ are located at distance $2T(n) + 2$ away from each other, $R_2$ will never meet $R_1$ within $T(n)$ rounds, so we can safely ignore how $R_2$ behaves. Now, after $T(n)$ rounds, $R_1$ will terminate. 

We subsequently have the adversary focus on $R_2$ and trap the robot within a strip of two nodes. Consider two adjacent nodes $u$ and $v$ and let $R_2$ be present on one of them. The adversary always removes the edge from the node that is not the edge between $u$ and $v$. Thus, $R_2$ will either terminate or indefinitely move between these two nodes.

In the course of the execution, $R_1$ could explore at most $T(n) + 1$ nodes before it terminates and $R_2$ could similarly explore at most those many nodes before being trapped. Thus, no more than $2T(n) + 2$ nodes could ever be explored, resulting in $\mathcal{A}$ failing. This is a contradiction and thus we see that no such $\mathcal{A}$ can exist.
\end{proof}

The reason the above impossibility result works for explicit termination but not for partial termination is that when we allow robots to use randomness to make choices, it no longer becomes clear which robot terminates first. This is not an issue for explicit termination because we leverage the fact that both robots eventually terminate and consider that running time. However, for partial termination, when utilizing the adversary to mimic the local view of one of the robots, it is unclear which robot we should focus on initially. And since we cannot focus on both simultaneously, if we pick the incorrect robot initially, we cannot guarantee that it will eventually terminate, and thus cannot move on to focus on the other robot.

We note that, the above proof strategy and observation can be extended to multiple robots when the adversary is made more powerful. Define a {\em t-adversary} as one which can remove at most $t$ edges in the graph in the given round. Note that, Observation~\ref{obs:never-meet} holds for $t+1$ robots starting at unique positions on a ring of size at least $t+2$. Thus, we can use a similar proof strategy to prove the following theorems.\footnote{We briefly recap the strategy. First run a supposed exploration algorithm $\mathcal{A}$ on a ring of size $n$ that terminates in $T(n)$ rounds. Subsequently, construct a ring of size $(t+1)(2T(n) + 2)$ and place the $t+1$ robots equidistant from each other. Now, for partial termination (explicit termination), run $\mathcal{A}$ on this larger ring and simulate the execution on the smaller ring for $1$ robot ($t$ robots) until it settles down and subsequently trap the remaining $t$ robots ($1$ robot) on already explored nodes. The total explored number of nodes will fall short the total size of the ring and hence $\mathcal{A}$ is incorrect.}

\begin{theorem}
There does not exist any exploration algorithm with partial termination of anonymous rings of unknown size at least $t+2$ by $t+1$ robots in the presence of a $t$-adversary, even when robots have distinct IDs, common chirality, the edge crossing detection capability, and when the scheduler is fully synchronous.
\end{theorem}

\begin{theorem}
There does not exist any exploration algorithm with explicit termination of anonymous rings of unknown size at least $t+2$ by $t+1$ robots in the presence of a $t$-adversary, even when robots have distinct IDs, common chirality, the edge crossing detection capability, access to randomness, and when the scheduler is fully synchronous.
\end{theorem}

\section{Deterministic Exploration with 3 Robots}
\label{explore}
In Section~\ref{sec:impossibility}, we showed that it is impossible to explore an anonymous and unknown size dynamic ring with two robots and achieve explicit termination. 
In this section, we present a deterministic solution for this exploration problem using three robots. We assume that each robot has a unique ID which is not known to the other robots unless they meet. We further assume that when two robots cross an edge (from opposite directions) in the same round, they {\em sense} each other and the meeting happened.\footnote{Here by `sense' we mean the two robots can detect the edge crossing and can exchange information including IDs.} 
Note that this edge crossing detection assumption does not help two robots (with unique IDs) to solve the exploration problem (see Section~\ref{sec:impossibility}). 
The outline of the algorithm is discussed below. The pseudocode is given in Algorithm~\ref{alg:explore3}.


The algorithm works in four stages: (Stage~1) first meeting of two robots, (Stage~2) second meeting of two robots, (Stage~3) exploration detection, and (Stage~4) termination. 

\textbf{Stage~1} ensures the first meeting of any two robots at some node  or via  edge crossing in the ring. For this, we need to make sure that at least two of them move in the opposite direction; otherwise if all the three robots move in the same direction at the same speed, they may never meet even if the adversary never deletes any edge. Thus we have to break this symmetry deterministically. For this, each robot moves based on the bit string of its ID. Each robot moves in phases and each phase consists of several rounds. More precisely, the number of rounds in the $i$-th phase is $2^i$. Without loss of generality, say that a robot moves in what it considers the clockwise (left) direction in phase $i$ when $b_{i \bmod k} = 0$. When $b_{i \bmod k} = 1$, the robot moves in the other (right) direction.
The first stage ends when at least two robots meet. Let us mark or name the two robots $A$ and $B$, where the larger ID one is $A$ and the other is $B$. Note that, the third robot may not know about this meeting and hence is unaware of the end of the first stage. Let us call the third robot as $C$.\footnote{The third robot gets named $C$ only after it meets either $A$ or $B$ at the end of Stage~2.} If these three robots meet at the same time (at some node) then the smallest ID robot gets named $C$. Notice that if two or three robots are positioned at the same node initially then the algorithm starts from Stage~2. 

Then \textbf{Stage~2} starts (which is known to at least $A$ and $B$). The robots $A$ and $B$ start moving in opposite directions from the meeting point (node) from Stage~1 and never change their directions until they terminate the algorithm. $A$ and $B$ each maintain a counter which counts the number of steps the robot successfully moves. Furthermore, each robot stores the ID of the other. Note that a robot cannot move in a particular direction in a round if the corresponding edge is missing (i.e., deleted by the adversary).  
Each of the robots continues to move until one of them meets the third robot. The second stage ends when either of $A$ and $B$ meets the third robot, which subsequently gets named $C$. Without loss of generality, assume that $A$ and $C$ meet. Then $A$ shares the following stored information with $C$: ID of $B$, the direction of $B$'s movement and the number of steps $A$ has successfully moved after Stage~1.\footnote{Note that even if $A$ and $C$ do not have shared chirality, the direction of $B$ can be conveyed as follows. Depending on how $A$ and $C$ meet, $C$ will immediately know the direction $A$ moves in or can take a round or two to understand this based on how $A$ and $C$ both move in their ``clockwise" direction and see if they moved to the same node or not. Once $C$ determines the naming mechanism $A$ uses for directions, $C$ can understand exactly which direction $B$ is moving in.} $C$ stores all this information. 
$A$ and $C$ also store each other's IDs.    

Then \textbf{Stage~3} starts, which ensures the completion of the ring exploration by at least two robots. This can happen in two ways. (I) if $A$ and $B$ meet (again) then it is guaranteed that exploration of the ring is complete. This scenario is depicted in Fig.~\ref{app:fig:meetAB}. (II) The adversary can prevent the meeting of $A$ and $B$ by removing an edge between them. Recall that $A$ and $B$ are moving in opposite directions. Eventually these two robots will reach two adjacent nodes and may wait for the (missing) edge to move. In this scenario, exploration is completed but $A$ and $B$ do not know this as $n$ is unknown. If the adversary does not remove the edge in one round, then $A$ and $B$ will meet. Therefore, the adversary will need to remove the edge indefinitely. In this situation, robot $C$ is used to determine the completion of exploration. From the meeting point of $A$ and $C$ in Stage~2, robot $C$ starts moving towards the opposite direction of $A$ (i.e., in the same direction of $B$) and $A$ continues moving in its fixed direction. Robot $B$ does not know that $A$ and $C$ met in Stage~2, and continues to move in its fixed direction. Robot $C$ moves towards $B$ until it catches $B$. Subsequently, $C$ changes its direction and move towards $A$ until it catches $A$. $C$ then repeats this process and moves back to $B$. Essentially, $C$ performs a zig-zag movement between $A$ and $B$ and checks if the distance (i.e., the hop distance) from $A$ to $B$ and $B$ to $A$ are the same. For this, the robot $C$ maintains two variables $AtoB$ and $BtoA$. $AtoB$ stores the number of successful steps (moves) towards $B$, starting from $A$ until it meets $B$, and $BtoA$ stores a similar number.
 When these two distances are equal, i.e., $AtoB = BtoA$, the algorithm determines that exploration is complete, as this condition implies that $A$ and $B$ lie on adjacent nodes whose edge has been removed by the adversary and thus $C$ has explored the entire graph. 
This scenario is depicted in Fig.~\ref{app:fig:exmode}. Therefore, either $A$ and $B$ meet and detect that  exploration is completed, or $C$ deduces the completion from the hop-distance counts. In the latter case, $C$ would be co-located in a node with either $A$ or $B$ and can thus inform that robot of the completion of exploration. Thus, at least two robots detect the exploration completion but the third robot may be unaware of this. Then we begin Stage~4 to ensure that all robots are made aware of exploration completion and can thus terminate.

\begin{figure}
    \centering
    \includegraphics[scale = 0.45]{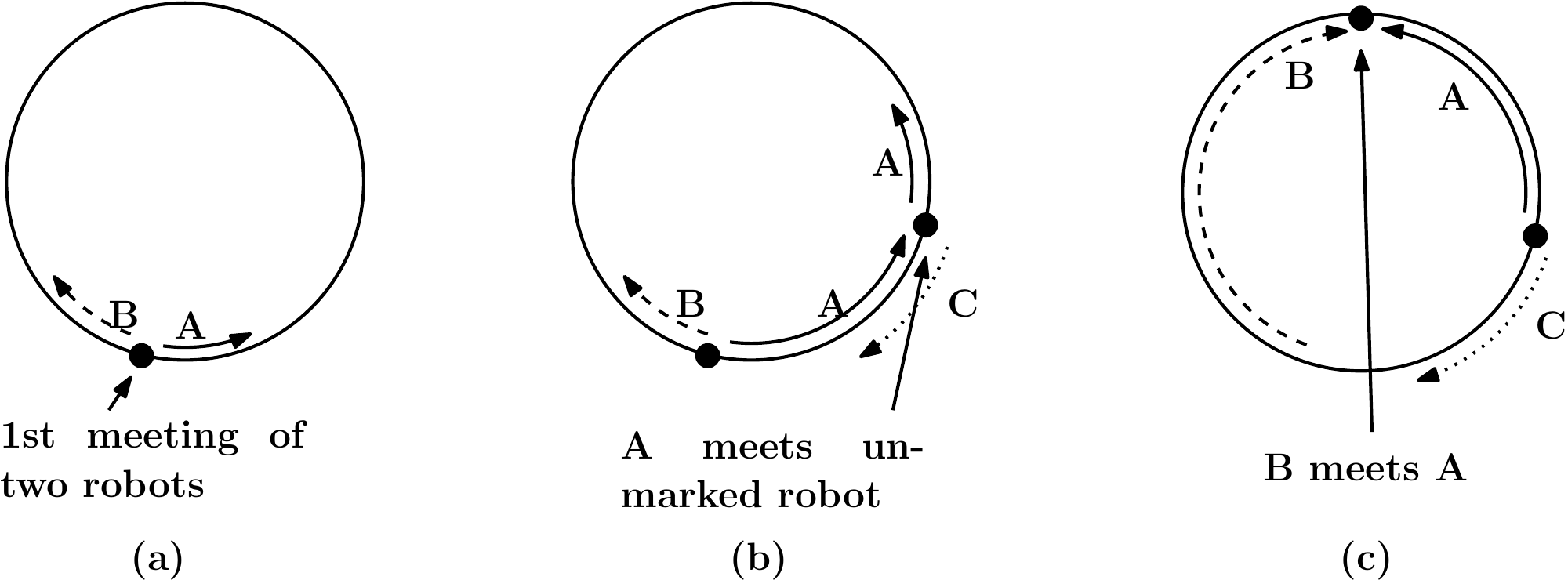}
    \caption{Here continuous, dashed and dotted lines shows the movement of robot $A$, $B$ and $C$ respectively. (a) Shows the scenario when two unmarked robot meet and get marked as $A$ and $B$, i.e. Stage 1 ends and Stage 2 starts. (b) Shows the scenario when $A$ and the unmarked robot meet and the unmarked robot gets marked as $C$, i.e. Stage 2 ends Stage 3 starts, (c) Shows the scenario when $A$ and $B$ meets, i.e. Stage 3 ends.}
    \label{app:fig:meetAB}
\end{figure}

\begin{figure}
    \centering
    \includegraphics[scale = 0.45]{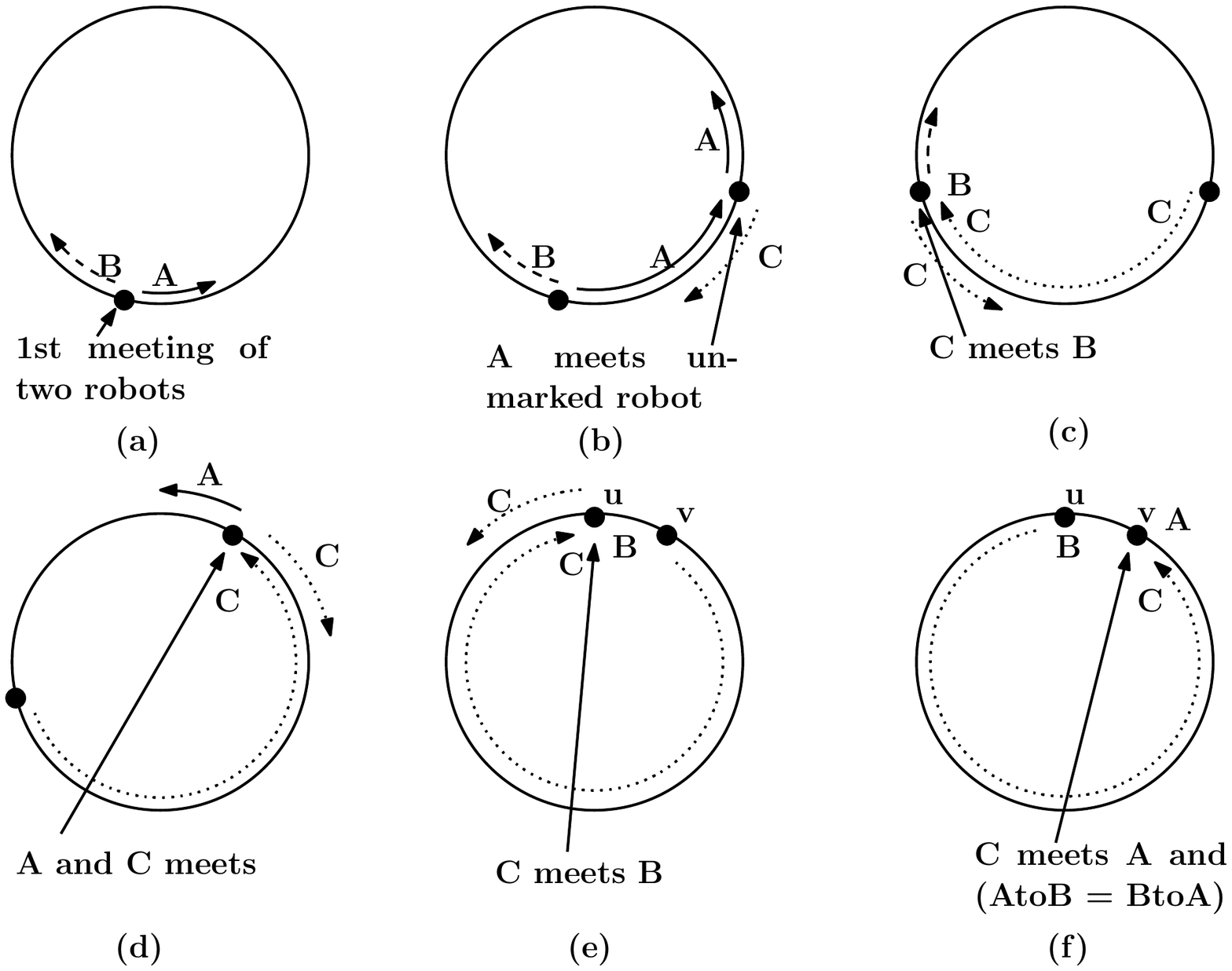}
    \caption{Here continuous, dashed and dotted lines shows the movement of robot $A$, $B$ and $C$ respectively. (a) Shows the scenario when two robots meet for the first time and get marked as $A$ and $B$, i.e. Stage 1 ends and Stage 2 starts. (b) Shows the scenario when $A$ and the 3rd unmarked robot meet and the unmarked robot gets marked as $C$, i.e. Stage 2 ends Stage 3 starts, (c) Shows the scenario when $C$ meets $B$ and finds $AtoB \neq BtoA$, (d) Shows the scenario when $C$ again meets $A$ at node $v$ and finds $AtoB \neq BtoA$ (e) Shows the scenario when $C$ again meets $B$ at node $u$ and finds $AtoB \neq BtoA$ (f) Shows the scenario when $C$ again meets $A$ at node $v$ and finds $AtoB = BtoA$ and Stage 3 ends. $u$ and $v$ are two consecutive nodes in $\R$. As nodes in $\R$ are not identifiable, the name of nodes are used for the ease of understanding.}
    \label{app:fig:exmode}
\end{figure}

In \textbf{Stage~4}, the two robots which detected the completion of exploration move in order to inform the third robot about the completion, and all robots terminate (to guarantee the explicit termination). Recall that the robots $A$ and $B$ maintained a counter of their successful moves starting from their meeting in Stage~2. If $A$ and $B$ meet at the same node again, then the sum of their counters is exactly $n$. If they meet by crossing each other, then the sum of their counters is exactly $n+1$. If they do not meet but the scenario from Stage~3 plays out, then $C$ and one of the two robots knows the value of $n-1$. 
In any case, there are two robots which know that exploration is completed and know the value of $n$ at some node. Then these two robots start moving in opposite directions to each other for at most $n$ rounds. When one of them meets the third robot, it informs the third robot about the completion of exploration and they both terminate. The other robot also terminates after $n$ rounds (if it does not meet the third robot). We show that after $n$ rounds, at least one of the robots (which detects the completion of exploration) meets the third robot and informs it about the completion of exploration.

The formal pseudocode of the algorithm is given in Algorithm~\ref{alg:explore3}. The pseudocode is written for the situation where at most two robots start on the same node or meet at the same node, to aid readability. However, it is easy to modify the code to handle the case of 3 robots meeting at the same node or initially located at the same node. The pseudocode requires us to define some parameters. Each robot maintains the following variables: 
\begin{itemize}
    \item {\em size}: size of the ring (initialized to $\infty$).
    \item {\em step}: stores the number of successful steps or moves of a robot.
    \item {\em mark}: takes value in $\{A, B, C\}$. Initialized to $NIL$ for all the robots.
\end{itemize}
Different events are defined below, which may occur when two or more robots meet during the execution of the algorithm.  
\begin{itemize}
    \item {\em meetSmall:} when a robot meets with a smaller ID robot and {\em mark} = $NIL$ for both the robots. 
    \item {\em meetLarge:} when a robot meets with a larger ID robot and {\em mark} = $NIL$ for both the robots.
    \item {\em meetMark:} when a robot, whose $mark = NIL$, meets another robot whose $mark \neq NIL$.
    \item {\em meetX:} when a robot, whose $mark \neq NIL$, meets another robot whose $mark = X$ for $X\in \{A, B, C\}$.
    \item {\em meetTer:} when a robot meets another robot which is executing {\sc startTermination} procedure.
\end{itemize}
Notice that when {\em meetSmall} event occurs for one robot, then {\em meetLarge} occurs for the other robot at the same time. In the pseudocode some functions are used. Small description of those functions are given below.
\begin{itemize}
    \item {\sc Move($dir$)}: By executing this function a robot moves one step towards the direction specified in $dir$. Here value of $dir$ can be $left$ or $right$.
    \item {\sc getMark()}: When two robots meet, this function returns the value stored at variable $mark$ on the other robot.
    \item {\sc getSteps($mark$)}: When two robots meet, this function returns the value stored at $step$ on the other robot.
    \item {\sc askTerminate($omark$)}: When two robots meet, this function signals to start termination to the other robot for which $mark = omark$.
    \item {\sc recTerminate()}: When two robots meet, this function detects whether a robot has received a signal to start termination or not. This function returns $1$ if signal to start termination is received.
\end{itemize}

In a ring there are two directions to move at each node. A robot chooses a direction arbitrarily and calls that direction as {\em left} and the other direction as {\em right}. Since the robots do not share any common chirality, the left and right directions of a robot may differ from the other robots (initially). When two or more robots meet, they can decide upon the directions and share a common chirality. The following algorithm (Algorithm~\ref{alg:explore3}) is executed by all the robots in parallel.

\subsection{Correctness and Time Analysis}\label{sec:correctness-and-time}
We first discuss the correctness of the algorithm in the following lemma. 


\begin{lemma}\label{lem:correctness}
Algorithm \ref{alg:explore3} correctly explores the dynamic ring and guarantees explicit termination.  
\end{lemma}

\begin{proof} We show that by the end of Stage~3, the ring is explored (by at least two robots). Stage~3 ends when one of the following two cases occurs: (I) robots $A$ and $B$ meet (II) robot $C$ meets $A$ and finds $AtoB = BtoA$ (after zig-zag movement) or $C$ meets $B$ and finds $AtoB = BtoA$.

\noindent {\em Case~I:} Robots $A$ and $B$ meet again after their first meeting in Stage~1. Since they move in opposite directions and never change directions after their first meeting, it is obvious that when they meet again, exploration is completed. The robots can also calculate the value of $n$ when they meet again. 

\noindent {\em Case~II:} Robot $C$ meets either of $A$ and $B$, and learns that $AtoB = BtoA$. This implies that robot $C$ has traversed same number of steps in two consecutive zig-zag movements. This scenario is only possible if both $A$ and $B$ are trying to traverse the same edge from adjacent nodes, since the adversary can remove only one edge at a time. Thus when robot $C$ determines that $AtoB = BtoA$ after two consecutive zig-zag movements, it is guaranteed that $C$ has explored all nodes in the ring and the size of the ring is $AtoB + 1$ or $BtoA + 1$. Fig.~\ref{app:fig:exmode} depicts this scenario. 

Thus, in both cases, at least two robots detect that exploration is completed. Moreover, the robots which detect this also know the size of the ring $n$ at the end of Stage~3. Thus in the termination stage, these two robots, which detected the completion of exploration, start moving in two opposite directions for at most $n$ rounds and terminate. It follows from the proof of Lemma~\ref{lem:meeting} (below) that after $n$ rounds, at least one of them meets the 3rd robot. So the 3rd robot also gets the information of the completion of exploration and terminates. Thus, explicit termination is guaranteed at the end of Stage~4.    
\end{proof}

Let us now analyze the time complexity of the exploration algorithm. We calculate the time taken in each stage of the algorithm.  

\begin{lemma}
\label{lem:meeting}
If two among the three robots move in opposite directions in a dynamic ring of size $n$, then at least two of them meet in at most $n - 2$ rounds. 
\end{lemma}

\begin{proof}
If any two among these three robots are initially located at the same node in the ring, then this lemma holds trivially. 

We assume that the three robots are initially located on three different nodes in the ring. In this scenario, we prove this lemma by induction on the size of the ring. Recall that a robot always tries to move in some specified direction so long as the corresponding edge is available, i.e. it does not voluntarily remain stationary. 

\noindent \textbf{Base case:} Consider the ring of size $3$ with each robot initially located on a different node. As the adversary can remove only one edge and at least two robots are moving in opposite directions, at least two of the robots will be on the same node after $1$ step (either the two moving in opposite directions meet or if the edge between them is not available then the third one catch one of them). Thus, the base case is true.

\noindent \textbf{Inductive step:} Assume that the claim holds on rings up to size $l$. We show that the claim also holds for rings of size $l+1$.  
In a ring of size $l+1$, when all three robots are located at different nodes, the maximum min-distance between two robots is at most $l-1$.\footnote{Here, min-distance refers to the shortest path distance between two robots on the ring.} As none of the robots are changing their direction and the adversary can remove only one edge in a round, at least one robot can successfully move one step in one round. Hence, after one round the situation on the ring of size $l+1$ maps to a situation on the ring of size $l$ (or at least two of them meet). As this lemma holds for a ring of size $l$ or less, at least two robots meet at some node after at most $l-2$ rounds. Therefore, it takes at most $(l-2)+1$, i.e., $l-1$ rounds, in a ring of size $l+1$. Hence the lemma holds.  
\end{proof}

\begin{lemma}
\label{lem:phase}
In Stage~1 of the algorithm, there exists a phase $i \in [0, k-1]$ when at least one robot moves in the direction opposite to the direction followed by other two robots, where $k$ is the length of the ID bit-string of the robots.  
\end{lemma}

\begin{proof}

Recall that in the $i^{th}$ phase a robot moves in some direction for $2^i$ rounds. It then changes its direction of movement in the next phase iff the next bit in its ID is different. Consider the scenario where all robots move in the same direction starting from the first phase (otherwise the lemma is trivially true). There are two scenarios to consider. Either all robots share the same chirality or they do not.

Consider the scenario where all robots share the same chirality. Since the IDs of the robots are different, at least one bit in the ID of each pair of the robots are different. Hence there will be at least two phases in between $0$ to $k-1$ when one of the robots moves in a direction opposite to the direction followed by the other two robots.

Consider the scenario where robots do not share the same chirality. Since chiralities can be different, two robots with different chiralities can have IDs that are complementary (e.g., 000 and 111) and thus move in the same direction in all phases. However, since all IDs are different, the third robots ID will be such that there will exist at least one phase where one of the robots moves in a different direction from the other two.
\end{proof}

\begin{lemma}
\label{lem:lruntime}
Stage~1 of Algorithm~\ref{alg:explore3} finishes in at most $n + n\cdot 2^k$ rounds, where $k$ is the length of the ID bit-string of the robots. 
\end{lemma}

\begin{proof}
We show that there exists a phase $i \in [0, j(k-1)]$ when at least two robots meet in Stage~1, where $k$ is the length of the IDs of the robots and $j$ is some positive integer. It follows from Lemma~\ref{lem:phase} that there exists a phase $i \in [0, k-1]$ when at least two robots move in the opposite directions to each other. The number of rounds in that phase is $2^i$. If $2^i \geq n-2$ then it follows from Lemma~\ref{lem:meeting} that any two robots meet. However, it might be the case that $2^i < n-2$ for that $i$ in $[0, k-1]$. 
Then according to our algorithm (see Stage~1), these two robots (again) move in the opposite directions in each of the phases $j(k-1)+i$ for $j =1 , 2, \dots$, and hence Stage~1 finishes when $2^{j(k-1)} \leq n-3$ and $2^{j(k-1) + i} \geq n-2$ for some positive integer $j$. Thus, Stage~1 takes at most $\sum_{t = 0}^{j(k-1) + i} 2^t$ rounds. The sum is bounded above by $(n + n\cdot 2^i$), since $2^{j(k-1)} \leq n-3$. Therefore, Stage 1 of Algorithm~\ref{alg:explore3} finishes in at most $n + n\cdot 2^k$ rounds, since $i < k$.      
\end{proof}
%

\begin{algorithm}[H]
 \footnotesize
 \caption{{\sc Explore-Dynamic-Ring-3-Robots}}
 \label{alg:explore3}
 \begin{algorithmic}[1]
 \State $i$ := $0$
 \While{($1$)}
    \If{$b_{i \bmod k}$ $=$ $0$}
        \State {\sc Explore}($left$, i)
    \ElsIf{$b_{i \bmod k}$ $=$ $1$}
       \State {\sc Explore}($right$, i)
    \EndIf
    \State $i$ := $i+1$
 \EndWhile
 \end{algorithmic}
\end{algorithm}

\begin{algorithm}[H]
 \footnotesize
 \floatname{algorithm}{Procedure}
 \caption{{\sc Explore}($dir$, i)}
 \label{alg:fexplore}
 \begin{algorithmic}[1]
    \For{$tstep = 0$ \textbf{to} ($2^i - 1$)}
        \If{{\em meetSmall}}             
            \State {\sc BeRobotAB($A$)}   
        \ElsIf{{\em meetLarge}}          
            \State {\sc BeRobotAB($B$)}
        \ElsIf{{\em meetMark}}           
            \State {\sc BeRobotC()}
        \EndIf
        \State {\sc Move($dir$)}
    \EndFor 
 \end{algorithmic}
\end{algorithm}
%

\begin{lemma}
\label{lem:stage2}
Stage~2 of Algorithm~\ref{alg:explore3} finishes in at most $n$ rounds. 
\end{lemma}
\begin{proof}
Stage~2 finishes when any of the two robots $A$ and $B$ meets $C$ after their ($A$ and $B$) first meeting in Stage~1. Note that $A$ and $B$ start moving in the opposite directions from after Stage~1. Thus it follows from the proof of Lemma~\ref{lem:meeting} that one of them meets $C$ in at most $n-1$ rounds.      
\end{proof}

\begin{lemma}\label{lem:stage3}
Stage~3 of Algorithm~\ref{alg:explore3} finishes in at most $4n$ rounds.
\end{lemma}
\begin{proof}
Stage~3 finishes when either (I) $A$ and $B$ meet again which ensures that exploration is completed, or (II) robot $C$ detects the completion of exploration from the step counts of its zig-zag movement, i.e., when $AtoB = BtoA = n-1$. Generally, the adversary can block the movement of one robot, since it can delete at most one edge. It can block the movement of two robots only when they are on the adjacent nodes of an edge and want to move through that edge, but the adversary keeps the edge deleted. Suppose these robots $A$ and $B$ could not meet in $n$ rounds (after the end of Stage~2). Then they must be at the adjacent nodes of the deleted edge, which follows from Lemma~\ref{lem:meeting}. Subsequently, robot $C$ will eventually find one of the robots, say $A$, in at most $n$ rounds. $C$ will then perform one iteration of the zig-zag movement and find $B$ in $n-1$ rounds (and set its counter $AtoB$ to $n-1$). In the next $n-1$ rounds, $C$ moves back to $A$ and sets $BtoA$ to $n-1$ as well. Thus, when $AtoB = BtoA = n-1$, $C$ concludes that exploration of the ring is complete and can inform $A$ of the same. Thus Stage~3 finishes in at most $4n$ rounds.
\end{proof}

\begin{lemma}\label{lem:stage4}
Stage~4 of Algorithm~\ref{alg:explore3} finishes in at most $n$ rounds.
\end{lemma}
\begin{proof}
In Stage~4, two robots, which have detected the completion of exploration, move for at most $n$ rounds and terminate. Since these two robots move in opposite directions, it follows from Lemma~\ref{lem:meeting} that at least one of them meets the 3rd robot in $n-2$ rounds and they both terminate. Hence Stage~4 takes $n$ rounds. 
\end{proof}

\begin{algorithm}[H]
\footnotesize
\floatname{algorithm}{Procedure}
\caption{{\sc BeRobotAB($recMark$)}}
\label{alg:makeA}
\begin{algorithmic}[1]
\State step := $0$
\If{$mark = NIL$}
    \State $mark := recMark$
\EndIf 
\While{($1$)}
    \If{$mark$ = $A$}
       \State {\sc Move($left$)}
    \Else   
       \State {\sc Move($right$)}
    \EndIf   
    \If{move successful}
       \State step := step+1
    \EndIf
    \If{{\em meetB} \textbf{or} {\em meetA}}  
       \If{$mark$ = $A$}
           \State recStep := {\sc getSetps($B$)}
           \State size := step + recStep + $1$
           \State {\sc startTermination}($A$, $left$, size)
       \Else
           \State recStep := {\sc getSetps($A$)}
           \State size := step + recStep + $1$
           \State {\sc startTermination}($B$, $right$, size)
       \EndIf
    \ElsIf{{\em meetC} \textbf{and} {\sc recTerminate()} = 1}    
       \State size := {\sc getSteps($C$)} + $1$
       \If{$mark$ = $A$}
          \State {\sc startTermination}($A$, $left$, size)
       \ElsIf{$mark$ = $B$}
          \State {\sc startTermination}($B$, $right$, size)
       \EndIf      
    \ElsIf{{\em meetTer}}                 
       \State {\em terminate}
    \EndIf
\EndWhile
\end{algorithmic}
\end{algorithm}
Now we state the main result of this section.  

\begin{theorem}\label{thm:main-3robots}
Algorithm \ref{alg:explore3} correctly explores a 1-interval connected dynamic (anonymous) ring of size $n$ in $O(n + n\cdot 2^k)$ rounds with $3$ robots such that each robot has unique ID of length $k$ bits and the robots have no knowledge of $n$ and no common chirality.
\end{theorem}

\begin{proof}
The correctness of the algorithm follows from Lemma~\ref{lem:correctness}. 

The running time of the algorithm follows from the time complexity analysis of the four stages in Lemmas~\ref{lem:lruntime},~\ref{lem:stage2},~\ref{lem:stage3} and~\ref{lem:stage4}. Thus by summing up the individual runtimes, we get the time complexity as
$(n + n\cdot 2^k) + n + 4n + n = 7n + n\cdot 2^k$. 


 Hence, Algorithm~\ref{alg:explore3} explores a dynamic ring of size $n$ with three robots in $(7n + n\cdot 2^k)$ rounds, where each robot has a unique ID of length $k$ bits. 
\end{proof}

\begin{algorithm}[H]
\footnotesize
 \floatname{algorithm}{Procedure}
\caption{{\sc BeRobotC()}}
\label{alg:makeC}
\begin{algorithmic}[1]
\State $AtoB$ := $0$
\State $BtoA$ := $0$
\State $recMark$ := {\sc getMark()}
\If{$mark = NIL$}
    \State $mark := C$
\EndIf
\If{$recMark$ = $A$}
    \State $dir$ := $right$
\ElsIf{$recMark$ = $B$}
    \State $dir$ := $left$
\EndIf
\While{($1$)}
    \State {\sc Move($dir$)}
    \If{move successful}
       \State step := step+1
    \EndIf
    \If{{\em meetB}}
       \State AtoB := step
       \If{$AtoB \neq BtoA$}
          \State step := 0
       \Else                                   
          \State size := step + $1$
          \State {\sc askTerminate()}
          \State {\sc startTermination}($C$, $left$, size)
       \EndIf
    \ElsIf{{\em meetA}}
       \State BtoA := step
       \If{$AtoB \neq BtoA$}
          \State step := 0
       \Else                                 
          \State size := step + $1$
          \State {\sc askTerminate()}
          \State {\sc startTermination}($C$, $right$, size)
       \EndIf
    \ElsIf{{\em meetTer}}                     
       \State {\em terminate}
    \EndIf   
\EndWhile
\end{algorithmic}
\end{algorithm}
\begin{algorithm}
\footnotesize
\floatname{algorithm}{Procedure}
\caption{{\sc startTermination}($mark$, $dir$, size)}
\label{alg:terminate}
\begin{algorithmic}[1]
\State Ttime := size
\While{(1)}
    \State {\sc Move($dir$)}
    \State Ttime := Ttime-1
    \If{Ttime = $0$}                             
        \State {\em terminate}
    \ElsIf{({\em meetA} \textbf{or} {\em meetB} \textbf{or} {\em meetC})}   
        \State {\em terminate}
    \EndIf
\EndWhile
\end{algorithmic}
\end{algorithm}

\begin{corollary}\label{cor:main-result}
There exists an algorithm which explores a 1-interval connected dynamic (anonymous) ring of size $n$ in $O(n)$ rounds with $3$ robots having unique IDs of length $O(1)$ bits and without the knowledge of $n$ and without common chirality.
\end{corollary}

\subsection{Simulation Results}\label{sec:experiment}
We perform experimental evaluation and highlight the effectiveness  and  efficiency  of our algorithm in dynamic rings for different parameter ranges. In particular, we evaluated the performance of our algorithm by computing the running time for different sizes of the dynamic ring (i.e., number of nodes in the ring) and also for different ID lengths of the robots. In the simulations, we assumed the robots are placed at random nodes in the beginning. Furthermore, each robot decides its initial direction of movement randomly, i.e. each robot decides clockwise (left) or anti-clockwise (right) direction randomly.  


We assume an adversary determines the dynamic ring in each round. In particular, the following four different adversarial strategies are considered for the simulations. 

\begin{itemize}
    \item \textit{Random Edge Deleted (RED):} At each round, the adversary randomly selects an edge in the ring and deletes it. The previously deleted edge gets added to the ring.
    \item \textit{Same Edge Deleted (SED):} The adversary randomly selects an edge in the ring and keeps the edge deleted throughout the execution.
    \item \textit{Random Robot Blocking (RRB):} At each round, the adversary targets a random robot, and block the movement of the robot in that round. This is done by deleting the edge through which the robot decides to move in that particular round. 
    \item \textit{Same Robot Blocking (SRB):} The adversary randomly selects a robot and block the movement of the robot throughout the execution by deleting appropriate edges. That is, the robot is not allowed to move from its initial position. 
\end{itemize}

The robots have no knowledge about the adversarial strategies, but the adversary knows  the robots' current position including the edges through which the robots decide to move. Thus the above adversarial strategies are adaptive. In all the cases, the dynamic ring remains connected throughout the execution.  

\noindent \textbf{Varying the Size of the Ring (Fig.~\ref{plot:ring_size}):}
In this experiment, we consider dynamic ring of five different sizes, i.e., $n = 20000, 40000, 60000, 80000$ and $100000$. The robots have ID of length $3$; in fact, the ID-bits are taken $100$, $101$ and $111$ for the 3 robots. We run the algorithm for $5$ times for each values of $n$ and count the average number of rounds taken to explore the ring. We plot the results in Fig.~\ref{plot:ring_size}, where the x-axis represents the ring size--$n$ and y-axis represents the obtained rounds to explore the ring. Observe that the time or the number of rounds to explore the ring increases when the size of the ring $n$ increases. However, in all the cases, the running time of our algorithm is bounded above by $5n$ for all the different adversarial strategies. The running time is less than $3n$ for the strategies RED and RRB. This shows that the simulation results outperform our theoretically proven time bound-- $7n + n\cdot 2^k$, where $k$ is the ID length of a robot. 

\begin{figure}
    \includegraphics[scale = 0.23]{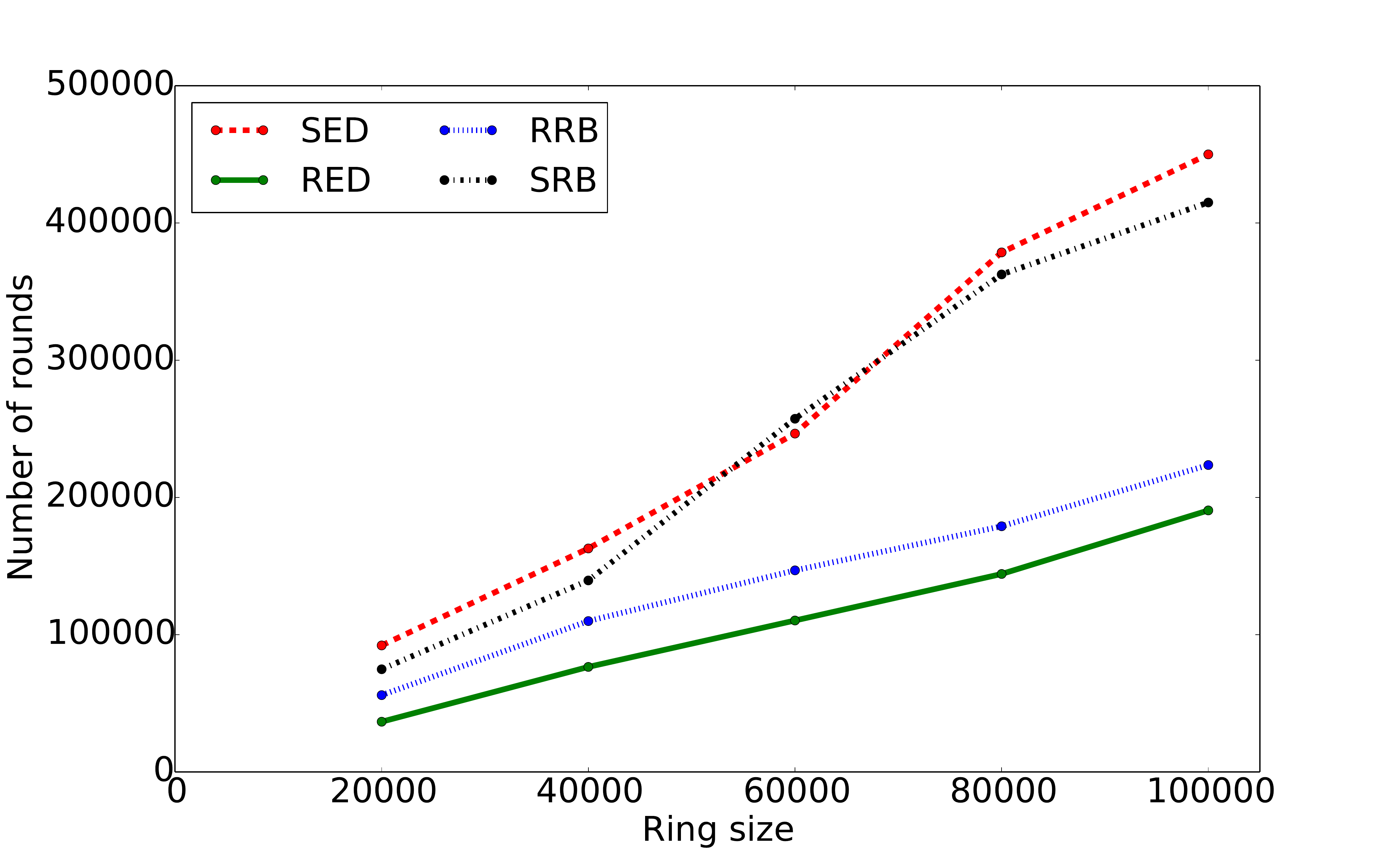}
    \caption{Varying $n$, the size of the ring.}
    \label{plot:ring_size}
\end{figure}

\noindent \textbf{Varying the ID Length of the Robots (Fig.~\ref{plot:ID_len}):}
In this experiment, we consider a ring of size $32768$. The length of the IDs of each robot varies from $3$ to $15$ (notice that $\log n = 15$ as $32768 = 2^{15}$). Particularly, ID length of the robots are considered $3$, $6$, $9$, $12$ and $15$ for the simulations. The ID-bits are generated randomly. In this case also, we run the algorithm for $5$ times on each different ID lengths and count the average number of rounds taken to explore the ring. Again we perform the simulations for the four adversarial strategies. We plot this simulation results in Fig.~\ref{plot:ID_len}, where the x-axis represents the ID length and y-axis represents the rounds taken by the algorithm to explore the ring. Observe that in all the cases, the running time of the algorithm is bounded above by $5n$, where $n$ is the ring size. In fact, the running time is less than $3n$ for the strategies RED and RRB. 
In this case also, the simulation results outperform our theoretically proven time bound.

\begin{figure}
    \includegraphics[scale = 0.23]{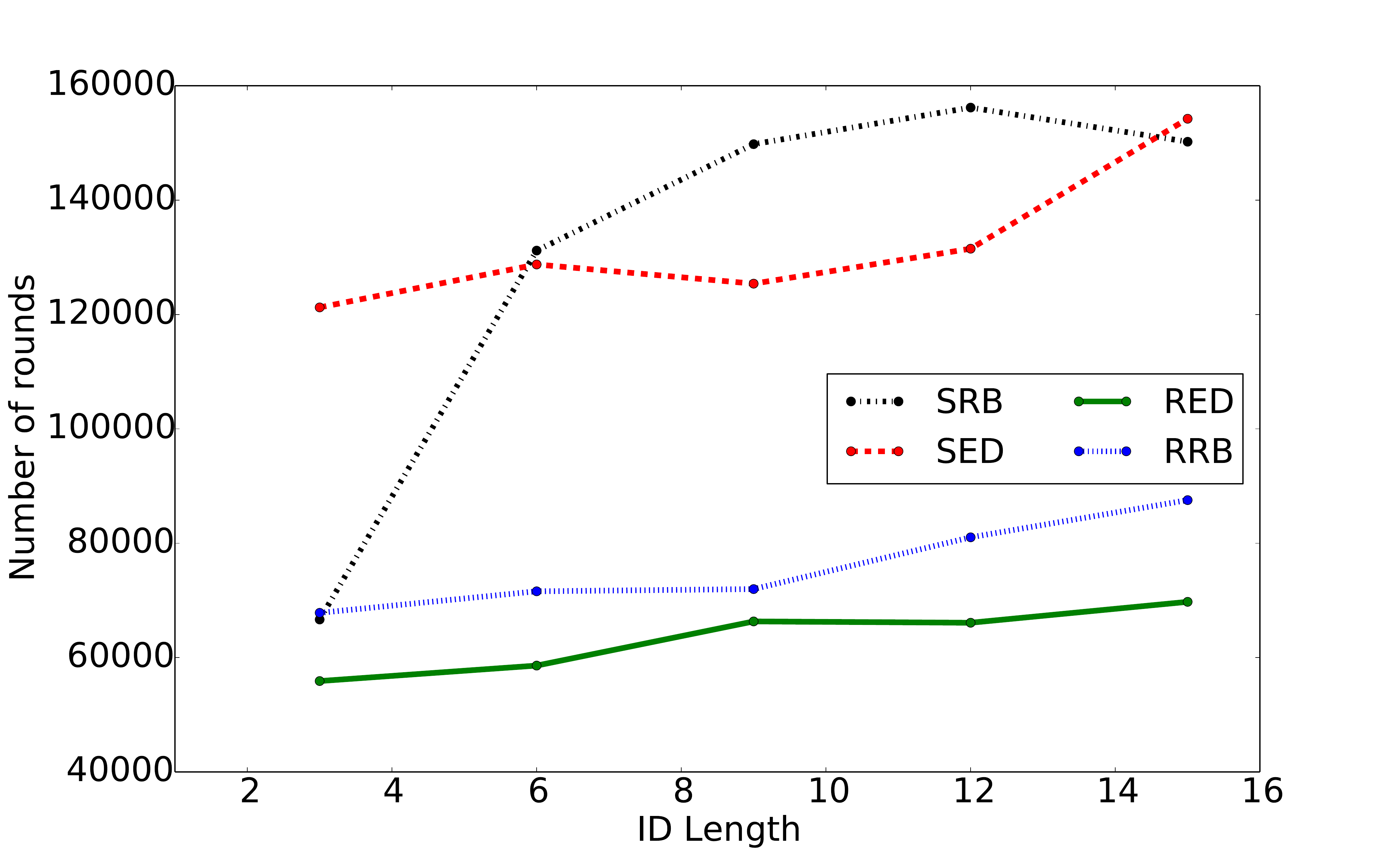}
    \caption{Varying the ID length of the robots.}
    \label{plot:ID_len}
\end{figure}

\sloppy
\section{Exploration with Randomness}
\label{sec:randomness-exploration}

We can remove the need for edge crossing detection through the creative use of randomness, resulting in an algorithm that is both Las Vegas and Monte Carlo in nature. We bound the algorithm's expected running time and success probability in the following section. Subsequently, we show how to further use randomness to remove the requirement of having unique IDs initially assigned to each robot.

\subsection{Removing Edge Crossing Detection}
In this section, we first discuss the changes to the Algorithm~\ref{alg:explore3} that are required to make it work. We then subsequently provide bounds for the running time and correctness.

Notice that edge crossing detection is used when two robots are located at adjacent nodes and must move in opposite directions along the same edge in the same round. A simple way to get the robots to meet is to force one to be stationary while the other moves. This can easily be achieved by having each robot flip a fair coin to decide if it should move or not. If the result of the coin toss is heads, then the robot performs the movement it initially planned to do. If the result is tails, then the robot does not move. Call this subroutine as {\sc Random-Movement}.

{\sc Random-Movement} can be run as a subroutine by every robot in every round, after deciding to move (and where to) but before the actual movement. We now describe how to further modify {\sc Explore-Dynamic-Ring-3-Robots} (Algorithm~\ref{alg:explore3}), in addition to using the subroutine, so that we may remove the need for edge crossing detection.

In Stage~1, notice that if the algorithm required a robot to move for $s$ steps in $r$ rounds, then directly using {\sc Random-Movement} in the algorithm may cause that robot to move for less than $s$ steps in $r$ rounds, even when the adversary does not block any movement of the robot. This is slightly problematic as we require that there exists a phase in which robots can move for at least $n$ steps, and we do not want to drastically increase the running time. An easy fix is to extend the number of rounds of each phase $i$ by a constant factor, say $8$, in order to ensure that on expectation and with high probability, the number of steps a robot moves through is at least $2^i$ when $2^i \geq n$.\footnote{The expectation is easy to see. The high probability bound can be seen by applying a simple Chernoff bound.} 

For Stage~2 to occur, we need two robots to come into contact with each other and be marked $A$ and $B$. Since the probability that two robots that were supposed to move through the same edge meet instead is $1/2$, on expectation, at least one such meeting occurs within $2$ phases of the first phase $i$ where $2^i \geq n$. Thus, on expectation after $O(n)$ rounds are complete, Stage~1 is over and Stage~2 begins.

Now, in Stage~2, we require either $A$ or $B$ to come into contact with the third robot. Notice that this third robot is not constrained to only move in one direction, but may move in both directions. Thus, it may only come into contact with either $A$ or $B$ as a result of crossing an edge. Again, through the use of {\sc Random-Movement}, it takes $2$ such attempts on expectation at edge crossing between the third robot and either $A$ or $B$ before contact is made and the third robot gets marked as $C$. Then the algorithm moves to Stage~3. Notice that before the third robot becomes marked, it is possible that $A$ and $B$ meet again, thus sending the algorithm directly into Stage~4 (since $A$ or $B$ may miss to meet $C$ as there is no edge crossing detection).

If the algorithm is in Stage~3, then the third robot has been marked $C$. Now, either $A$ and $B$ meet again or the adversary blocks $A$ and $B$ at adjacent nodes and $C$ moves back and forth between them. In the latter case, it takes $O(n)$ rounds with high probability\footnote{The high probability is a result of the use of {\sc Random-Movement}.} to move to Stage~4 and $C$, and one of $A$ and $B$ will know the exact value of $n$.\footnote{It is possible that $A$ and $B$ may have crossed each other several times before the adversary blocks them at adjacent nodes and the latter case occurs. In this case, $C$ and the robot it finally interacts with will know an upper bound on $n$. Note that if $A$ and $B$ cross each other at least twice, then on expectation they will meet, leading to the former case.} In the former case, after some $cn$ rounds, where $c$ is a positive constant, $A$ and $B$ will meet on expectation. Thus, $A$ and $B$ will know an upper bound $cn$ of $n$.

In Stage~4, let us assume that without loss of generality, two of the robots $A$ and $B$ learn an upper bound $N$ on the value of $n$, i.e., $N = cn$ for some constant $c$. Now, either the third robot is marked or it is not. Either way, our goal in this stage is to inform this third robot that exploration is complete. It is possible that every interaction of $A$ or $B$ with the third robot is a situation where edge crossing would normally occur. In the event of one such interaction, the probability of the third robot being informed is $1/2$. After $2 \log N$ such interactions, the probability of the third robot being informed is at least $1 - 1/n$.  Recall that in this stage, robots $A$ and $B$ use a counter and will stop after $N$ rounds. If we change the counter to end at $16N \log N$ instead, then the third robot has $2N \log N$ opportunities with high probability to interact with either robot. Thus, with probability at least $1 - 1/n$ the third robot will interact with at least one of the robots and terminate, eventually resulting in explicit termination.

Construct the new algorithm {\sc Modified-Explore-Dynamic-Ring-3-Robots} using the above mentioned modifications to the stages and the use of {\sc Random-Movement}.

\begin{theorem}\label{the:randomness-alg}
When the robots run {\sc Modified-Explore-Dynamic-Ring-3-Robots}, exploration of the ring with explicit termination occurs with probability at least $1 - 1/n$ in $O(n \log n)$ rounds on expectation.
\end{theorem}

\begin{proof}[Proof Sketch]
The running time is a result of the modifications to Stage~4 adding $O(n \log n)$ rounds and the use of {\sc Random-Movement} changing exact running time to expected running time. 

The change from explicit termination to explicit termination with high probability is a result of the modification to Stage~4. As we have two robots surely terminating after a certain number of rounds, we can only guarantee with probability at least $1 - 1/n$ that the third robot will terminate as well.
\end{proof}

\begin{remark}\label{rem:partial-termination}
If we relax our termination condition to partial termination, we can eliminate Stage~4 of {\sc Modified-Explore-Dynamic-Ring-3-Robots} entirely and simply have robots terminate instead of moving to Stage~4. Then the exploration of the ring with partial termination occurs with high probability in $O(n)$ rounds on expectation.
\end{remark}



\subsection{Assigning Unique IDs}\label{subsec:unique-id-randomness}
Throughout this paper, we made the assumption that each robot was initially assigned a unique ID from the range $[1,2^k]$ prior to the start of the algorithm, where $k$ is the length of the ID bits. This assumption can be removed by having each robot pick an ID uniformly at random from a range of numbers $[1,2^l]$, where $l$ is a parameter to the algorithm\footnote{Since the number of robots $3$ is given to the robots, they can use it to set a value for $l$, e.g., $l=2^3+12 =20$.}. It is easy to see that the probability that all robots have unique IDs is $(1-1/2^l)(1-2/2^l) = 1 - O(1/2^l)$. It should be noted that although $l$ can be made arbitrarily large to improve the probability that each robot has a different ID, a larger value of $l$ possibly results in a longer runtime of the algorithm (refer to Lemma~\ref{lem:lruntime}).

\begin{theorem}\label{the:unique-id-removal}
When robots run {\sc Modified-Explore-Dynamic-Ring-3-Robots} and choose IDs uniformly at random from the range $[1,2^l]$, exploration of the ring with explicit termination occurs with probability at least $(1 - 1/n)(1 - O(1/2^l))$ in $O((n + n\cdot 2^l) \log n)$ rounds on expectation. 
\end{theorem}


\section{Exploration in Semi-Synchronous Setting (SSYNC)}\label{sec:ssync}
In this section, we show how to extend our ideas to the passive transport semi-synchronous model proposed in~\cite{DDFS16} in order to achieve exploration with partial termination using 3 robots even in the absence of a landmark or the knowledge of $n$. Recall that in the semi-synchronous setting, in every round, a subset of the robots are put to sleep by the adversary with the restriction that the number of rounds any robot remains asleep is finite. In this setting, passive transport relates to how the robot moves given the following setup. Suppose a robot is awake and wants to travel along an edge $e$ in a round $i$ and the adversary removed that edge in that round. Now, suppose the adversary subsequently puts the robot to sleep from round $i+1$ until some round $j$. If $e$ is present again for the first time in some round $k:\, i+1 \leq k \leq j$, then the robot moves along the edge in round $k$ even though it is asleep. 

In this setting, Di Luna et al.~\cite{DDFS16} showed that with 3 robots and either the knowledge of an upper bound on $n$ or the presence of a landmark, they were able to achieve exploration with partial termination. We show that it is possible for 3 robots to achieve exploration with partial termination without either of the above two requirements, so long as the robots have access to randomness and the ability of edge crossing detection. We assume that robots have unique IDs, but that requirement can be removed through the use of randomness as described in Section~\ref{subsec:unique-id-randomness}.

We consider the four stage algorithm presented in Section~\ref{explore} and show how to modify it to achieve exploration with partial termination in this passive transport semi-synchronous model. We first replace Stage~1 with the following zero round protocol. Each robot flips a coin and chooses which direction to move (until Stage~2 is reached) based on the result of the coin toss. With probability $3/4$, two robots will move in one direction while the third moves in the other direction. Stage~2 proceeds as described in the original algorithm. Stage~3 is modified as follows. Robot $C$ will check for an additional condition before determining that the ring has been explored. If $AtoB = BtoA$, and $A$ and $B$ were both trying to move on an edge removed by the adversary, then $C$ determines that the ring has been explored. We explain below how $C$ can detect that $A$ and $B$ were trying to move on an edge. Note that it is not necessary that $A$ and $B$ were awake when $C$ visited, but merely that they were attempting to move.

The reasoning behind the above changes is that Stage~1 and Stage~4 require robots to rely on counting the number of rounds. While the simple trick of flipping coins to choose directions solves the Stage~1 problem, there is no immediate fix to the problems present in Stage~4. For Stage~3, we require the above change in order to protect against the adversary simply putting $A$ and $B$ to sleep while $C$ moves back and forth between them. The condition ensures that $C$ has to see them both wanting to move (but not necessarily awake) and prevented to by the adversary before deciding to terminate. Since the adversary can only keep a robot asleep for a finite number of rounds and only remove at most one edge from the graph, eventually, one of the robots will make progress on the ring until either the condition is met or $A$ and $B$ meet.

There is the following subtlety to take into account. Suppose that two robots cross the same edge in the same round (or end up co-located at the same node) and at least one of them is asleep. We need both of them to detect that such an edge crossing (or meeting) occurred and furthermore, be able to swap data with one another. This data should include information about whether one of the robots tried to move along an edge while awake but was subsequently put to sleep before the move could be completed.

With the above modifications, we get an algorithm with the following properties.

\begin{theorem}
There exists an algorithm that correctly explores a $1$-interval connected dynamic (anonymous) ring of size $n$ with probability $3/4$ in $O(fn)$ steps, where $f$ is the largest interval of time between two consecutive activations of any robot, using $3$ robots with unique IDs that neither have common chirality, nor knowledge of an upper bound on $n$, nor access to a landmark.
\end{theorem}

Note that we measure number of steps moved and not running time. Furthermore, note that the number of steps is $O(fn)$ and not $O(n)$. This is due to the fact that the adversary can put $A$ and $B$ to sleep for an arbitrarily long time in Stage~3, but not an infinitely long time.

\section{Conclusions}
\label{sec:conclusions}

In this paper, we looked into the problem of exploration of a dynamic ring in the presence of 1-interval connectivity. We first showed that exploration with explicit termination subject to some constraints with just two robots equipped with unique IDs even with access to edge crossing detection and randomness is impossible. Subsequently, we presented a deterministic algorithm where three uniquely identifiable robots with edge crossing detection capability explore any 1-interval connected dynamic ring in optimal time. We also showed how to remove the requirement of this capability and allow the robots to be anonymous while still achieving explicit termination with high success probability through the use of randomness. We finally extended our results to the semi-synchronous setting.

There is an interesting line of future research. Our algorithms intimately used advance knowledge of the number of robots present in the system. If that knowledge is unknown and $\geq 3$ robots are present, is there an algorithm to solve exploration with explicit termination?



\bibliographystyle{IEEEtranS}
\bibliography{references}


\end{document}